\documentclass[twocolumn,superscriptaddress,showpacs,prx]{revtex4-2}
\usepackage[latin9]{inputenc}
\pdfoutput=1
\usepackage[colorlinks=true,allcolors=blue]{hyperref}
\usepackage{orcidlink}

\usepackage{booktabs}   
\usepackage{siunitx}    

\usepackage[ruled,vlined]{algorithm2e}
\usepackage{dsfont}

\usepackage{amsmath}
\usepackage{amssymb}
\usepackage{amsthm}

\pdfoutput=1
\usepackage[english]{babel}
\usepackage[T1]{fontenc}

\usepackage{amsthm}  
\usepackage{amsmath}
\usepackage{mathtools}
\usepackage{bbold}

\newcommand{\Z}{\phantom{-}}

\def\bra#1{\langle#1|} \def\ket#1{|#1\rangle}

\def\ketbra#1#2{\ket{#1}\!\bra{#2}}

\DeclareMathOperator\Tr{Tr}

\newtheorem{theo}{Theorem}

\newtheorem{prop}[theo]{Proposition}

\newtheorem{mainresult}{Main Result}

\def\bra#1{\langle#1|} \def\ket#1{|#1\rangle}

\def\ketbra#1#2{\ket{#1}\!\bra{#2}}

\newcommand{\st}{\mathrm{s.t.}}
\newcommand{\id}{\mathds{1}}

\def\one{\leavevmode\hbox{\small1\normalsize\kern-.33em1}}

\newcommand{\Perm}[1]{\mathcal P{#1}}

\newcommand{\sign}{\text{sgn}}

\begin{document}

\title{Bound entanglement-assisted prepare-and-measure scenarios based on four-dimensional quantum messages}

\author{Istv\'an M\'arton\,\orcidlink{0000-0001-7024-8245}}
\affiliation{HUN-REN Institute for Nuclear Research, P.O. Box 51, H-4001 Debrecen, Hungary} 

\author{Erika Bene\,\orcidlink{0000-0002-8073-7525}}
\affiliation{HUN-REN Institute for Nuclear Research, P.O. Box 51, H-4001 Debrecen, Hungary} 

\author{Tam\'as V\'ertesi\,\orcidlink{0000-0003-4437-9414}}
\affiliation{HUN-REN Institute for Nuclear Research, P.O. Box 51, H-4001 Debrecen, Hungary}

\begin{abstract}
We present a class of linear correlation witnesses that detects bound entanglement within a three-party prepare-and-measure scenario with four-dimensional quantum messages. We relate the detection power of our witnesses for two-ququart Bloch-product-diagonal states to that of the computable cross norm-realignment (CCNR) criterion. Several bound entangled states in four or even higher dimensions, including those which are useful in metrology, can exceed the separable bound computed by reliable iterative methods. In particular, we show that a prominent two-ququart bound entangled state with a positive partial transpose (PPT) can be mixed with up to $40\%$ isotropic noise and still be detected as entangled by our prepare-and-measure witness. Furthermore, our witnesses appear to be experimentally practical, requiring only the use of qubit rotations on Alice's and Bob's sides and product qubit measurements with binary outcomes on Charlie's side.
\end{abstract}

\maketitle

\section{Introduction}

The ability to prepare, transform and measure entangled quantum states is of fundamental importance in quantum information science, which forms the basis for the development of quantum technologies~\cite{Guhne2009,Horodecki2009}. One of the most puzzling forms of quantum correlations -- bound entanglement -- was first introduced by the Horodecki family in 1998. These are correlations that cannot be produced by local operations and classical communication (LOCC). On the other hand, no pure entanglement can be extracted from them by LOCC operations, even if arbitrary copies are available~\cite{Horodecki1998mixed}. Nevertheless, it has been shown that bound entanglement is useful and provides a resource for quantum information processing. For instance, it is relevant in quantum many-body systems~\cite{Toth2007} and has proven valuable in several applications, including quantum key distribution~\cite{Horodecki2005,Ozols2014}, quantum metrology~\cite{Czekaj2015,Toth2018,Pal2021}, channel discrimination~\cite{Piani2009}, activation of quantum resources~\cite{Horodecki1999bound,Masanes2006,Tendick2020,Toth2020activating}, Einstein-Podolsky-Rosen (EPR) steering~\cite{Moroder2014steering} and Bell nonlocality~\cite{Vertesi2014}. For a recent review of bipartite bound entanglement, see Ref.~\cite{Hiesmayr2024}. 

However, in the case of Bell correlations~\cite{Brunner2014Bell}, the amount of Bell violation obtained using bipartite bound entangled states with positive partial transpose (PPT)~\cite{Peres1996} has so far been very small~\cite{Vertesi2014,Yu2017,Pal2017family}. In contrast, bipartite PPT states can generate quantum correlations in dimension-bounded prepare-and-measure (PM) scenarios, leading to significant violations of classical bounds. This surprising property was recently demonstrated by Carceller and Tavakoli~\cite{Carceller2025PRL} in a PM setup similar to the one introduced in Ref.~\cite{Bakhshinezhad2024}.

In fact, in prepare-and-measure scenarios with dimension-bounded quantum messages, entanglement enables correlations that are much stronger than those achievable without shared entanglement~\cite{Tavakoli2021correlations}. An emblematic example is superdense coding~\cite{Bennett1992communication}, where the communication capacity of a single qubit can be doubled if the sender and receiver share a two-qubit entangled state. This protocol has recently been reformulated within the prepare-and-measure framework~\cite{Moreno2021,Tavakoli2018semi}, and a distributed three-party variant has also been proposed~\cite{Bowles2015testing}. 

The PM construction presented in Refs.~\cite{Bakhshinezhad2024, Carceller2025PRL} has two key properties. First, it is a three-party black-box communication protocol involving two senders and one receiver. Second, in addition to a shared entangled resource between the senders, it enables them to transmit $d$-dimensional quantum messages (where $d$ is prime or an odd prime) to the receiver, rather than classical messages.

Both features appear crucial to obtain strong violations of the classical-separable bounds with only weakly entangled states. In contrast, in a bipartite PM setup despite its entanglement signaling power there are indications that it still cannot close the so-called Werner gap. In particular, there exist a range of mixedness in Werner's family~\cite{Werner1989} for which the state remains entangled yet is undetectable by any possible set of adaptive product measurements in the bipartite PM scenario~\cite{Bakhshinezhad2024}. 

In this work, we introduce a class of correlation witnesses specifically tailored to detect PPT bound entanglement with high robustness to noise within a three-party prepare-and-measure scenario. While our construction shares some features with the witness of Carceller and Tavakoli~\cite{Carceller2025PRL}, notably, the use of dimensionally bounded quantum messages in a three-party PM setting and product measurements, it also differs from theirs in several key aspects.

Crucially, Carceller and Tavakoli~\cite{Carceller2025PRL} construct their witness for arbitrary odd-prime message dimensions $d$, providing exact results for $d=3,5,7$. In contrast, our witness works specifically for $d=4$. Remarkably, Carceller and Tavakoli using a numerical search can detect very weakly entangled $d\times d$ PPT states: their PM witness tolerates isotropic noise up to $18.83\%$, $28.53\%$, and $28.69\%$ for $d=3,5,7$, respectively. 

By reproducing the optimal states for $d=3,5,7$ from Ref.~\cite{Carceller2025PRL}, we verify that each can be written in the $d\times d$ generalized Bell-diagonal basis consisting of a set of $d^2$ maximally entangled states~\cite{Bennett1993,Sych2009}. Moreover, when the above stated isotropic noise levels are added to these PPT states, the resulting mixtures turn out to sit only marginally above the entanglement threshold witnessed by the computable cross norm-realignment (CCNR) criterion~\cite{Rudolph2005,Chen2003}. So, for these states, the CCNR criterion is only slightly more powerful in detecting entanglement than the Carceller-Tavakoli correlation witness. We also note that in this setup, which uses odd-prime $d$ dimensional messages, the receiver's station performs $d$-outcome measurements instead of the usual two-outcome ones. 

On the other hand, our correlation witness is specifically built for a $4$-dimensional message space, which can be tailored to target generic entangled states. This allows us to detect broad families of weakly entangled $4\times 4$ PPT states. By requiring only simple Pauli rotations and product Pauli measurements, it provides a versatile tool for entanglement detection under realistic, noisy conditions. Indeed, our reliable numerical analysis reveals that every $4\times 4$ Bloch-product-diagonal PPT state identified by the CCNR criterion is also caught by our witness. By contrast, for generic $4\times 4$ states, our witness matches the efficiency of the so-called trace criterion. These observations are formalized as the central findings of this work in Main Result~\ref{theorem1}. In addition, unlike the PM witness of Ref.~\cite{Carceller2025PRL}, our individual Pauli measurements do not require post-processing of the outcomes depending on the input $z$; Charlie simply deterministically wires his outcomes to produce the final output $c=\pm 1$.

Furthermore, we show that our construction can detect bound entangled states in any dimension from $D=3$ up to $D\to\infty$, and can even outperform the CCNR criterion for $D>4$. However, as opposed to the Carceller-Tavakoli witness~\cite{Carceller2025PRL}, which generalizes to arbitrary odd-prime dimensions (modulo the conjectured separable bound for $d>7$), our method does not directly extend to arbitrary message dimensions, including odd primes or their products. This limitation arises from the special group-theoretic properties underlying our design, which we examine in detail in the following sections.

\section{Scenario}

We consider a three-party black-box scenario in which the quantum messages are restricted to dimension four ($d=4$). Let the three parties, Alice, Bob and Charlie select their inputs $x$, $y$, and $z$, respectively. Each of these inputs is a pair of numbers, where Alice's input is given by $x=(x_0,x_1)\in\{0,1,2,3\}^2$. Similarly, Bob's and Charlie's inputs are $y=(y_0,y_1)$ and $z=(z_0,z_1)$. We define an index $x=4x_0+x_1+1$, so that $x$ runs from 1 to 16, enumerating the pairs ($x_0,x_1$) in lexicographic order. Similar indexing is used for Bob's and Charlie's input. In our setup, Charlie's measurement is binary for each input $z$ producing an output $c=\{\pm 1\}$. Fig.~\ref{fig1setup} provides an illustration of the scenario.

In the entanglement based model, the probability that Charlie obtains outcome $c$ given the inputs $x$, $y$, and $z$ is given by
\begin{equation}
P(c|x,y,z)=\Tr(\rho_{AB}^{(x,y)}M_{c|z}),
\end{equation}
where the transformed states by Alice and Bob conditioned on $(x,y)$ are
\begin{equation}
\label{rhoabxy}
\rho_{AB}^{(x,y)}
\;=\;
\bigl(\mathcal{E}_x\otimes \mathcal{F}_y\bigr)\bigl(\rho_{AB}\bigr),
\end{equation}
where $\mathcal{E}_x$ on $A$ and $\mathcal{F}_y$ on $B$ are a pair of local completely positive trace-preserving (CPTP) maps into a four-dimensional message space. The Hermitian operators $M_{c|z}$ are the elements of Charlie's positive operator-valued measure (POVM) that add up to the identity, $\sum_{c=\pm1}M_{c|z}=\id_4$ for every input $z$. Because the measurement is binary and $P(+1|x,y,z)=1-P(-1|x,y,z)$, it is useful to define the expectation value $E_{xyz}=P(+1|x,y,z)-P(-1|x,y,z)$. In the quantum scenario, the correlator can be expressed as
\begin{equation}
E_{xyz}=\Tr{\left[\rho_{AB}^{(x,y)}\,C_z\right]},    
\label{EQxyz}
\end{equation}
where the $4\times 4$ state $\rho_{AB}^{(x,y)}$ is given by Eq.~(\ref{rhoabxy}) and $C_z$ is Charlie's dichotomic observable (with eigenvalues $\pm 1$) for input $z$. This captures all the information which is available from a prepare-and-measure experiment with the topology shown in Fig.~\ref{fig1setup}. 

\begin{figure}[t!]
\includegraphics[trim=35 265 135 20,clip,width=12cm]{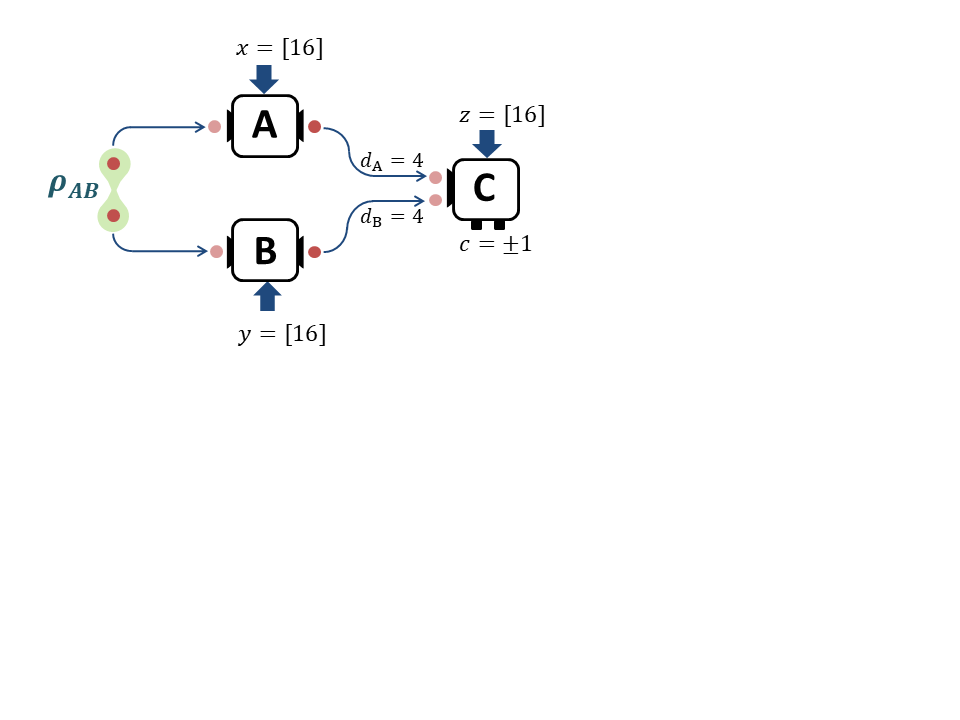}
\caption{Three-party prepare-and-measure scenario. When entanglement is allowed, Alice and Bob may share any bipartite quantum state $\rho_{AB}$. They encode their respective inputs $x,y\in\{1,\ldots,16\}$ into four-dimensional quantum messages ($d_A=d_B=4$), which they send to Charlie. Upon receiving his own input $z\in\{1,\ldots,16\}$, Charlie performs a measurement and outputs $c=\pm 1$. In the unentangled case, $\rho_{AB}$ is restricted to be separable, taking the form \eqref{rhosep}.} 
\label{fig1setup}
\end{figure}

We distinguish between two completely different cases based on the property of the quantum state shared by Alice and Bob. In the general case, $\rho_{AB}$ is an arbitrary bipartite state. In contrast, the unentangled or separable case is characterized by states that can be written as a convex mixture of product states~\cite{Horodecki2009,Guhne2009}
\begin{equation}
\rho_{AB}=\sum_k p_k\rho_k^A\otimes \rho_k^B, 
\label{rhosep}
\end{equation}
where the $p_k$ weights sum up to unity. A state that cannot be written in this form is entangled. To certify entanglement in the prepare-and-measure scenario, we make use of linear functionals built from correlators $E_{xyz}$. Specifically, we consider inequalities of the form 
\begin{equation}
    \sum_{x,y,z} w_{xyz}E_{xyz}\le Q_{\text{sep}},
    \label{wit}
\end{equation}
with some fixed real coefficients $w_{xyz}$, where the value $Q_{\text{sep}}$ is chosen so that the bound holds for all separable states~(\ref{rhosep}), for any choice of CPTP maps $\mathcal{E}_x$, $\mathcal{F}_y$ and observables $C_z$ in the correlator \eqref{EQxyz}. By linearity, it suffices to optimize over pure product states in $\mathbb{C}^4\otimes\mathbb{C}^4$. Hence the separable bound can be written as  
\begin{equation}
Q_{\mathrm{sep}}
=\max_{\rho_x^A,\rho_y^B, C_z} \sum_{x,y,z} w_{xyz}\Tr(\rho_x^A\otimes\rho_y^B C_z),
    \label{witsep}
\end{equation}
where $\rho_x^A$ and $\rho_y^B$ are arbitrary four-dimensional pure states.

In what follows, we construct a correlation inequality in the form of (\ref{wit}), where the witness coefficients $w_{xyz}$ are tailored to the specific entangled state we wish to detect. To achieve this, we first introduce a class of $4\times 4$ quantum states given by their Bloch-decomposition parameters. We then express the witness coefficients in terms of these parameters. We provide strong numerical evidence that any $4\times 4$ Bloch-diagonal state violates our correlation inequality whenever it violates the CCNR criterion. Furthermore, we extend our analysis to non-Bloch-diagonal states, where the CCNR value is replaced by the Manhattan norm of the diagonal elements of the correlation tensor of the state. Our central finding, encompassing these statements, is presented in Main Result~\ref{theorem1}.

\section{A class of Bloch-diagonal states}

A generic $4\times 4$ state $\rho_{AB}$ can be expanded in a Bloch-product-operator-basis as 
\begin{equation}
\rho_{AB}  = \sum_{k,\ell=1}^{16} t_{k\ell}\,A_{k}\otimes B_{\ell}\,,
\label{rhodecomp}
\end{equation}
where the real coefficients 
\begin{equation}
\label{tkl}
    t_{k\ell}=\Tr(\rho_{AB}A_k\otimes B_\ell)
\end{equation}
are the entries of the correlation matrix $T = (t_{k\ell})$. Above the operators $A_k$ are defined by 
\begin{equation}
\label{Ak}
A_k=\frac{1}{2}\sigma_{k_0}\otimes\sigma_{k_1},\quad k=(k_0,k_1)\in\{0,1,2,3\}^2,
\end{equation}
with $\sigma_0=\id$, $\sigma_1=X$, $\sigma_2=Y$, and $\sigma_3=Z$ being the Pauli matrices. On the other hand, one may choose the operators $B_k$ on $B$ to be a permutation of operators $A_k$ on $A$. Concretely, let $\Perm{}$ be a permutation of the index set $\{1,\dots,16\}$, and define $B_k = A_{\Perm(k)}, k=1,\dots,16$.

We consider a class of $4\times 4$ states that can be expressed in a Bloch-product-diagonal form (or Bloch-diagonal, for short) as
\begin{equation}
\rho_{AB}=\sum_{k=1}^{16}\lambda_k A_k\otimes B_k.\label{rhobloch}
\end{equation}
Note that $\lambda_k$ can be negative and in this case $t_{k\ell}=\lambda_k\delta_{k\ell}$ in the decomposition~(\ref{rhodecomp}). Note that the sets $\{A_k\}_{k=1}^{16}$ and $\{B_k\}_{k=1}^{16}$ form orthonormal bases of $4\times 4$ Hermitian operators. That is, 
\begin{equation}
\Tr{\left(A_kA_\ell\right)}=\Tr{\left(B_kB_\ell\right)=\delta_{k\ell}}.  
\end{equation}
In the special case of identity permutation, defined by $\Perm{(k)}=k$ for all $k$, we have $B_k=A_k$. The coefficients $\lambda_k$ are real (and may take negative values) and are chosen so that $\rho_{AB}$ is a valid quantum state, namely, it must be positive ($\rho_{AB}\succeq 0$) and normalized ($\Tr{(\rho_{AB})}=1$). For instance, if $\Perm{(1)}=1$, then $B_1=A_1$, then we have $\lambda_1=1/4$. Such Bloch-diagonal states have been discussed in Refs.~\cite{Moroder2012calibration,Sentis2018}. 

Note that any $d\times d$ state, in particular any $4\times 4$ state, can be written in its Schmidt form as
\begin{equation}
\rho_{AB}=\sum_{k=1}^{d^2}\tilde{\lambda}_k \tilde{A}_k\otimes\tilde{B}_k,
\label{rhoschmidt}
\end{equation}
with $\tilde{\lambda}_k>0$, where $\{\tilde{A}_k\}$ and $\{\tilde{B}_k\}$ are orthonormal bases of Hermitian operators satisfying $\Tr\bigl(\tilde{A}_k\,\tilde{A}_\ell\bigr)
= \Tr\bigl(\tilde{B}_k\,\tilde{B}_\ell\bigr) = \delta_{k\ell}$, and $\tilde{\lambda}_k=\Tr{(\rho_{AB}\tilde{A}_k\otimes\tilde{B}_k)}$. Even for $d=4$, the operators $\tilde{A}_k$ and $\tilde{B}_k$ need not factor into Pauli matrices. Moreover, the Bloch decomposition in Eq.~(\ref{rhobloch}) can be brought to a Schmidt form by absorbing the signs of the $\lambda_k$ into the corresponding $B_k$ operators. 

We now recall the so-called CCNR criterion for entanglement detection in $d\times d$ systems.

\begin{prop}[CCNR Criterion {\cite{Rudolph2005,Chen2003}}]
\label{prop1}
Let a bipartite state \(\rho_{AB}\) admit the Schmidt decomposition \eqref{rhoschmidt}. Define the CCNR value by $\mathrm{CCNR}(\rho_{AB})
\;\equiv\;\sum_{k=1}^{d^2} \tilde{\lambda}_k$.
For any separable state $\rho_{AB}$, we have 
$\mathrm{CCNR}(\rho_{AB}) \le 1$. Violation of this bound certifies that $\rho_{AB}$ is entangled.
\end{prop}

Applying Proposition~\ref{prop1} to the Bloch-diagonal family~(\ref{rhobloch}), the violation of the CCNR criterion $\sum_{k=1}^{16}|\lambda_k|\le 1$ detects entanglement in any state of the form~(\ref{rhobloch}). Although this Bloch-diagonal class does not encompass all possible $4\times 4$ entangled states, as we will see, it does include many with rich entanglement properties, including bound entangled states.

\begin{prop}[Trace Criterion {\cite[Sec.\,2.5.1]{Guhne2009}}]
\label{prop1b}
Let us express a $4\times 4$ state $\rho_{AB}$ in the decomposition~\eqref{rhodecomp}, where the coefficients $t_{k\ell}$ are given by \eqref{tkl}. Define the sum 
\begin{equation}
\label{trS}
S(\rho_{AB}) \;\equiv\; \sum_{k=1}^{16} \bigl|t_{kk}\bigr|.    
\end{equation}
Then for any separable state $\rho_{AB}$, we have $S(\rho_{AB}) \le 1$.  Consequently, violation of this bound certifies that $\rho_{AB}$ is entangled.
\end{prop}

This proposition follows directly from Proposition~\ref{prop1} together with a special case of von Neumann's trace inequality~\cite{Mirsky1975}, which for any $4\times 4$ state $\rho_{AB}$ gives
\begin{equation}
\label{JvN}
\mathrm{CCNR}(\rho_{AB})
= \sum_{k=1}^{16} \tilde{\lambda}_k
\;\ge\;
\sum_{k=1}^{16} \bigl|t_{kk}\bigr|=S(\rho_{AB}).    
\end{equation}
Accordingly, the sum of the singular values of the matrix $T=(t_{k\ell})$ is always at least the sum of the absolute values of its diagonal entries, with equality if and only if $\rho_{AB}$ is Bloch-diagonal (i.e.,\ \(t_{k\ell}=\pm\delta_{k\ell}\,\tilde{\lambda}_k\)). Note that $S(\rho_{AB})$ in Eq. \eqref{trS} is simply the Manhattan norm of the vector $\text{diag}(T)$.

Hence, while this trace-based criterion is generally weaker than the CCNR test for non-Bloch-diagonal states, it still detects a large portion of $4\times 4$ bound entangled states, as we will show. Moreover, for any fixed $\rho_{AB}$, one may further strengthen the entanglement detection by applying local unitaries on $\rho_{AB}$ to maximize $S$.

\section{Prepare-and-measure witnesses for message dimension four}

We now introduce a family of linear witnesses in the form of correlation inequalities~(\ref{wit}) which can be tailored to generic $4\times 4$ quantum states. In particular, these witnesses are optimally suited to the Bloch-diagonal class of states~(\ref{rhobloch}), as we shall see. 

Consider an arbitrary $4\times 4$ state $\rho_{AB}$ with decomposition~(\ref{rhodecomp}), where $t_{k\ell}$ and $A_k$ are defined by Eqs.~\eqref{tkl} and \eqref{Ak}, respectively, and $B_k=A_{\Perm{(k)}}$ for $k=1,\ldots,16$ under some fixed permutation $\Perm{}$. We then define the witness coefficients $w_{xyz}$ appearing in the correlation inequality~(\ref{wit}) by
\begin{equation}
   w_{xyz}
\,=\,
\operatorname{sgn}(t_{zz})
\Tr\bigl(A_x A_z A_x A_z\bigr)
\Tr\bigl(B_y B_z B_y B_z\bigr) 
\label{witcoeff}    
\end{equation}
for $x,y,z\in\{1,\ldots,16\}$. Here the sign function is defined as 
\begin{equation}
\sign(r) =
\begin{cases}
-1, & \text{if } r < 0 \\
\Z 0, & \text{if } r = 0 \\
+1, & \text{if } r > 0. 
\end{cases}    
\end{equation}    
Since $\left|\Tr{(A_xA_zA_xA_z)}\right|=\left|\Tr{(B_yB_zB_yB_z)}\right|=1/4$ holds for every choice of $x,y,z$, each coefficient $w_{xyz}$ takes on a value of either $+1/16$ or $-1/16$. 

In the subsections that follow, we show that the correlation witness \eqref{wit} with coefficients \eqref{witcoeff} can attain the value (see Prop.~\ref{prop3})
\begin{equation}
\label{Qent}
Q(\rho_{AB}) = 4^3 \sum_{z=1}^{16} \bigl|t_{zz}\bigr|,  
\end{equation}
whereas, in the absence of shared entanglement, the witness value is lower-bounded by (see Prop.~\ref{prop4})
\begin{equation}
Q_{\mathrm{sep}}^M = 4^3.
\label{QsepLB}
\end{equation}
In addition, using a reliable iterative algorithm, we conjecture that
\begin{equation}
Q_{\mathrm{sep}} \;\le\; 4^3    
\label{qsep}
\end{equation}
for every witness defined by \eqref{witcoeff}.  This bound can be saturated using four-dimensional unentangled quantum messages, either by pure product states or by mixtures of product states, as we shall demonstrate below. 

Moreover, noting that $Q(\rho_{AB})>Q_{\text{sep}}$ by Eqs. \eqref{trS}, \eqref{Qent} and \eqref{qsep} whenever $S(\rho_{AB})>1$, we formalize the main result of this paper in the following statement.

\begin{mainresult}
Consider a generic $4\times 4$ state expressed in the decomposition
$\rho_{AB}  = \sum_{k,l=1}^{16} t_{kl}\,A_{k}\otimes B_{\ell}$, where the real coefficients $t_{k\ell}=\Tr(\rho_{AB}A_k\otimes B_\ell)$, and the operators $A_k$ and $B_\ell$ are products of normalized Pauli matrices as in Eq. \eqref{Ak}. Then the correlation inequality $\sum_{xyz} w_{xyz}E_{xyz}\le Q_{\text{sep}}$ defined by the coefficients $w_{xyz}$ in Eq.~\eqref{witcoeff} can be violated whenever $S(\rho_{AB})=\sum_{k=1}^{16} \left|t_{kk}\right|>1$. 

Moreover, if $\rho_{AB}$ can be written in a Bloch-product-diagonal form, $\rho_{AB}=\sum_{k=1}^{16}\lambda_k A_k\otimes B_k$, then the correlation inequality $\sum_{xyz} w_{xyz}E_{xyz}\le Q_{\text{sep}}$ can be violated whenever $\text{CCNR}(\rho_{AB})>1$.
\label{theorem1}    
\end{mainresult} 

Note that Eqs.~\eqref{Qent}, \eqref{QsepLB}, and \eqref{qsep} form the basis of the main result above. Equations~\eqref{Qent} and \eqref{QsepLB} are proved in Sections~\ref{Qentsec} and \ref{Qsep}, respectively, while Eq.~\eqref{qsep}, supported by extensive numerical evidence, is demonstrated via an iterative algorithm described in Sec.~\ref{Qsep}.

\subsection{Entangled value}
\label{Qentsec}

To prove Eq.~\eqref{Qent} we take the local CPTP maps in Eq.~\eqref{rhoabxy} to be unitary channels generated by $U_x$ and $V_y$ on $\mathbb{C}^4$. Hence 
\begin{equation}
\label{rhoabxyU}
\rho_{AB}^{(x,y)}
\;=\; \bigl(U_x\otimes V_y\bigl)\rho_{AB}\bigl(U^{\dagger}_x\otimes V^{\dagger}_y\bigl).
\end{equation}

\noindent Then we have
\begin{prop}
\label{prop3}
Let $\rho_{AB}$ be expressed in the decomposition \eqref{rhodecomp} with coefficients $t_{k\ell}$ as in \eqref{tkl}. Choose the unitaries in \eqref{rhoabxyU} as
\begin{equation}
\label{UVC}
U_x = 2A_x,\quad
V_y = 2B_y,\quad    
\end{equation}
and set Charlie's two-ququart observables as 
\begin{equation}
C_z = 4\,A_z \otimes B_z.     
\end{equation}
Then the witness \eqref{wit} with coefficients \eqref{witcoeff} evaluates to
\begin{equation}
  Q(\rho_{AB}) \;=\; 4^3 \sum_{z=1}^{16} \bigl|t_{zz}\bigr|.
  \label{generic}
\end{equation}
\end{prop}

\begin{proof}
By $C_z = 4\,A_z \otimes B_z$, the expectation value in \eqref{EQxyz} becomes $E_{xyz}=4\Tr(\rho_{AB}^{(x,y)} A_z\otimes B_z)$. Substituting Eqs. \eqref{rhoabxyU},\eqref{UVC} and the expansion~(\ref{rhodecomp}) then gives
\begin{equation}
E_{xyz}
= 4^3 \sum_{k,\ell} t_{k\ell}\,
  \Tr\!\bigl[A_x A_k A_x A_z\bigr]\,
  \Tr\!\bigl[B_y B_\ell B_y B_z\bigr].    
\label{Exyztkl}  
\end{equation}
One then shows via the Pauli conjugation relations 
\begin{equation}
\sigma_i\,\sigma_j\,\sigma_i
=\begin{cases}
+\;\sigma_j, & i=j,\\
-\;\sigma_j, & i\neq j,
\end{cases}    
\label{pauliconj}
\end{equation}
that 
$|\Tr{(A_xA_kA_xA_z)}|=|\Tr{(B_yB_kB_yB_z)}|=(1/4)\delta_{k,z}$ for any $x,y$. Hence all off-diagonal terms vanish in Eq.~\eqref{Exyztkl} and we have
$E_{xyz}=4^3\times t_{zz}\Tr{(A_xA_zA_xA_z)}\Tr{(B_yB_zB_yB_z)}$. Finally, inserting the coefficients from Eq.~\eqref{witcoeff} into the witness $\sum_{xyz} w_{xyz}E_{xyz}$
yields $Q(\rho_{AB})=4^3\sum_{xyz}|t_{zz}|(\Tr{(A_xA_zA_xA_z)}\Tr{(B_yB_zB_yB_z)})^2=4^3\times\\\sum_z|t_{zz}|$. 
\end{proof}

\noindent Some observations follow.

\medskip

\noindent (i) Consider a Bloch-diagonal state $\rho_{AB}$ as in \eqref{rhobloch}.  In the decomposition \eqref{rhodecomp}, let us identify its diagonal components \(t_{zz}\) with coefficients \(\lambda_z\).  Then one immediately obtains
\begin{equation}\label{Qrhoab}
Q(\rho_{AB})
= 4^3 \sum_{z} \bigl|\lambda_{z}\bigr|
= 4^3 \times \mathrm{CCNR}(\rho_{AB})\,.
\end{equation}
Moreover, by \eqref{JvN}, for an arbitrary $4\times 4$ state one always has
\begin{equation}
\mathrm{CCNR}(\rho_{AB})
\;\ge\;
\sum_{z}\bigl|t_{zz}\bigr|,
\end{equation}
with equality precisely in the Bloch-diagonal case.

\medskip

\noindent (ii) Substituting $S(\rho_{AB})$ from Eq. \eqref{trS} into Eq.~\eqref{generic}, it follows that 
\begin{equation}
Q(\rho_{AB})=4^3\times\,S(\rho_{AB})\;\le\;Q_{\mathrm{ent}}.   
\end{equation}

\medskip

Furthermore, for a given $4\times 4$ state $\rho_{AB}$, the witness value may be driven closer to $Q_{\text{ent}}$ by optimizing over local basis rotations on Alice's and Bob's sides. Concretely, let
\begin{equation}
U'_x = U_A\,A_x,\qquad
V'_y = U_B\,B_y,    
\end{equation}
where $U_A$ and $U_B$ are local unitaries on Alice's and Bob's systems, respectively. Define the diagonal coefficients of the rotated correlation tensor by
\begin{equation}
t'_{zz}=\Tr\bigl[\rho_{\mathrm{AB}}\,(U_A A_z U_A^\dagger)\otimes(U_B B_z U_B^\dagger)\bigr].
\end{equation}
By maximizing $S'=\sum_{z=1}^{16}\left|t'_{zz}\right|$ over all choices of $U_A$ and $U_B$, one can potentially increase the achievable witness value.

\smallskip

\noindent Finally, one may ask whether more general encoding maps $\mathcal{E}_x$ for Alice and $\mathcal{F}_y$ for Bob, or more general observables $C_z$ for Charlie, can boost the witness value further. As we shall demonstrate, such enhancements are indeed possible when the dimension of the state space of $\rho_{AB}$ exceeds $4\times 4$.

\subsection{Separable value}
\label{Qsep}

With the coefficients~(\ref{witcoeff}) in the witness \eqref{wit} fixed by the target state~\eqref{rhodecomp}, our next task is to maximize the expression~(\ref{witsep}) over all two-ququart product states $\rho_x^A\otimes\rho_y^B$ and measurement observables $C_z$. Let us denote this maximum by $Q_{\text{sep}}$. First we establish the lower bound of $Q_{\text{sep}}^M=4^3$ as stated in Eq.~(\ref{QsepLB}). 

Let us first observe that in order to obtain $Q_{\text{sep}}$ for any witness of the form~\eqref{witcoeff} it is enough to consider a canonical witness with coefficients
\begin{equation}
    w_{xyz}=\Tr{(A_xA_zA_xA_z)}\Tr{(A_yA_zA_yA_z)}
\label{witcoeffcan}    
\end{equation}
for $x,y,z\in\{1,\ldots,16\}$, which is given by the parameters $\text{sgn}(t_{zz})=1$ and $B_z=A_z$ in \eqref{witcoeff}. This follows from the fact that, if we fix the witness by the coefficients \eqref{witcoeff} and assume that the product states $\rho_x^A\otimes\rho_y^B$ together with the measurement observables $C_z$ achieve $Q_{\text{sep}}$, then we can construct the permuted product states $\rho_x^A\otimes\rho_{\mathcal{P}(y)}^B$ and the measurement observables $\text{sgn}(t_{zz})C_z$ to achieve $Q_{\text{sep}}$ for the canonical witness \eqref{witcoeffcan}. Therefore, without loss of generality, we may now focus on determining the separable bound of this canonical witness. 

First we obtain the lower bound $Q_{\text{sep}}^M=4^3$ to $Q_{\text{sep}}$ by fixing Charlie's observables as $C_z=4A_z\otimes A_z$ (for $z=1\ldots,16$) and then maximizing over all pure ququart states $\rho_x^A$ and $\rho_y^B$. Next we carry out an extensive numerical search to support the claim in \eqref{qsep} that $Q_{\text{sep}}=Q_{\text{sep}}^M=4^3$, by optimizing over all possible choices of observables $\{C_z\}$.

\begin{prop} 
A lower bound for $Q_{\text{sep}}$ is given by $Q_{\text{sep}}^M = 4^3$, and this value can be attained using separable product states.
\label{prop4}
\end{prop}

\begin{proof}
Substitute the canonical witness coefficients~(\ref{witcoeffcan}) into the witness expression for separable states~(\ref{witsep}). We then obtain
\begin{equation}
\sum_{xyz}H_{xyz}\Tr{(2A_z\rho_x^A)}\Tr{(2A_z\rho_y^B)}, 
\label{witH}
\end{equation}
where 
$H_{xyz}=\Tr((A_zA_x)^2)\Tr((A_zA_y)^2)$. Since $|H_{xyz}| = 1/16$, we have that the expression~(\ref{witH}) is upper bounded by 
$(1/16)\sum_{xyz}\left|\Tr{(2A_z\rho_x^A)}\Tr{(2A_z\rho_y^B)}\right|$. Applying the Cauchy-Schwarz inequality and using the orthonormal and completeness properties of the observables~$\{A_k\}$, we arrive at the upper bound $ (1/16)\sum_{xy}\sqrt{\sum_z{(\Tr{(2A_z\rho_x^A)}})^2}\sqrt{\sum_z{(\Tr{(2A_z\rho_y^B)})^2}}\le (1/16)\sum_{x,y}(2\times 2)=4^3$.
Then, using Prop.~\ref{prop3}, the value $Q_{\text{sep}}^M=4^3$ can be recovered for any $4\times 4$ state $\rho_{AB}$ satisfying $S(\rho_{AB})=1$ in Eq.~\eqref{trS}. Such separable states, mixed or product, indeed exist, such a product state is $(\ket{00}\bra{00})^{\otimes 2}$.
\end{proof}

Although $Q_{\text{sep}}^M = 4^3 \le Q_{\text{sep}}$ by construction, numerical evidence indicates that equality holds. To find the tight bound~$Q_{\text{sep}}=4^3$ which holds for generic observables, we apply a see-saw iterative procedure~\cite{Werner2001,Liang2009,Pal2010} to optimize the canonical witness defined by the coefficients (\ref{witcoeffcan}). 

We find that, in addition to the observables $C_z=4A_z\otimes A_z$, the bound can also be obtained using measurements expressed in the standard basis. See the end of the section for further details on the corresponding optimal states and observables. For more on see-saw methods in the characterization of quantum correlations, see Ref.~\cite{Tavakoli2024semidefinite}. 
\\[6pt]
\noindent\textbf{A see-saw algorithm.}
The see-saw iterative method for optimizing a prepare-and-measure witness with fixed coefficients $w_{xyz}$ in Eq.~(\ref{wit}) over two-ququart separable states is as follows.
\begin{enumerate}
  \item Pick a set of random pure states $\{\rho_x^A\}_{x=1}^{16}$ and $\{\rho_y^B\}_{y=1}^{16}$ in Hilbert space of dimension four.
  \item Express the witness value as $w=\sum_{z=1}^{16}\Tr{(F_z C_z)}$, where $F_z=\sum_{xy}w_{xyz}\rho_x^A\otimes\rho_y^B$.
  \item Keeping $F_z$ fixed, maximize the objective value $w$ by optimizing over all operators subject to $-\id_4\le C_z\le\id_4$ for each $z=1,\ldots,16$. To do so, let $\sum_j{c_j^z\ketbra{\phi_j^z}{\phi_j^z}}$ be the singular value decomposition of $F_z$ and choose $C_z=\sum_j{\sign(c_j^z)\ketbra{\phi_j^z}{\phi_j^z}}$. 
  \item Rewrite the witness value as $w=\sum_x\Tr{(G_x \rho_x^A)}$, where $G_x=\sum_{yz}w_{xyz}\Tr_B{(\id_4\otimes \rho_y^B C_z)}$. Then, for each $x$, set $\rho_x^{\text{A}}=\ketbra{\psi_x}{\psi_x}$, where $\ket{\psi_x}$ is the eigenvector corresponding to the largest eigenvalue of $G_x$. 
  \item Similar to the previous step, update $\rho_y^B$ for fixed $\rho_x^A$ and $C_z$.
  \item Repeat steps 2-5 until convergence of the witness value $w$ in Eq.~(\ref{wit}) is reached.
\end{enumerate}

It is important to note that the see-saw method described above is heuristic and may get stuck in a local maximum, so the algorithm is typically repeated multiple times to ensure global optimality. Note that in our case, it was sufficient to apply the method only to the canonical witness with coefficients~\eqref{witcoeffcan}. By repeating the see-saw scheme, the algorithm has consistently converged to $w=4^3$. Similar see-saw-type algorithms have been used in various quantum correlation setups, including quantum metrology~\cite{Macieszczak2014,Toth2018,Toth2020activating,Kurdzialek2025}, EPR steering~\cite{Moroder2014steering}, prepare-and-measure scenarios~\cite{Tavakoli2017dim,Silva2023}, Bell nonlocality~\cite{Pal2010, Navascues2011}, and entanglement detection~\cite{Navascues2020genuine,Lukacs2022}, and usually a steady convergence of the algorithm has been observed.

With a slight modification of the above procedure, the witness can also be optimized in the case of a classical model. In this model, instead of quantum messages, Alice and Bob send four-dimensional classical messages to Charlie. The maximum witness value, $Q_{\text{c}}$ can be reproduced by using measurement observables $C_z$ that are diagonal in the computational basis with $\pm 1$ entries. Then the bound $Q_{\text{c}}$ can be computed using the same see-saw procedure, but by randomly initializing the pure states in step 1 to computation basis states, $\{\ket{i_{A,x}}\}_x$ and $\{\ket{i_{B,y}}\}_y$. 

In this way we find that the following symmetric classical deterministic strategies for Alice and Bob lead to the saturation of $Q_{\text{c}}=Q_{\text{sep}}=4^3$ in the canonical choice of the witness coefficients~(\ref{witcoeffcan}). Namely, Alice sends the message $m_A=x_1$ (encoded in the basis state $\rho_x^A=\ketbra{x_1}{x_1}$) to Charlie and Bob sends $m_B=y_1$ (encoded in the basis state $\rho_y^B=\ketbra{y_1}{y_1}$) to Charlie, where the inputs are $x=(x_0,x_1)$ and $y=(y_0,y_1)$.  If Charlie chooses an appropriate decoding strategy, as described in step 3 of the see-saw iterative method above, this procedure indeed yields $Q_{\text{c}}=Q_{\text{sep}}=4^3$. 

\subsection{Connection to the CCNR value}
\label{deeper}

It is natural to ask if there is a deeper conceptual reason behind the connection between the entanglement detection power of the correlation witness \eqref{wit} with coefficients given by \eqref{witcoeff}, and the CCNR criterion~(\ref{Qrhoab}), as expressed by Main Result~\ref{theorem1}. Here, we give such an explanation for the case of Bloch-diagonal states, where the $4\times 4$ state is written as $\rho_{AB}=\sum_z\lambda_z A_z\otimes B_z$, as in Eq.~(\ref{rhobloch}). Then, using the conjugation property of Pauli matrices~(\ref{pauliconj}), we observe that applying the unitaries $U_x=2A_x$ and $V_y=2B_y$ to the state $\rho_{AB}$ results in another state in Bloch-diagonal form:
\begin{equation}
\rho^{(x,y)}_{AB} = \sum_{z=1}^{16} s^{(x,y)}_z\lambda_z A_z\otimes B_z,    
\label{rhoxy}
\end{equation}
where $s_z^{(x,y)} = 16\Tr(A_xA_zA_xA_z)\Tr(B_yB_zB_yB_z)$ and take the values of $\pm 1$. Note that each state $\rho^{(x,y)}_{AB}$ has the same CCNR value of $\sum_z|\lambda_z|$. 

On the other hand, given a state in the Bloch-diagonal form \eqref{rhobloch}, the CCNR criterion can be associated with an entanglement witness operator~\cite{Guhne2006}
\begin{equation}
W_{\text{CCNR}} = \id_{16} - \sum_{z=1}^{16} \text{sgn}(\lambda_z) A_z \otimes B_z,     
\end{equation}
which ensures $\Tr(W_{\text{CCNR}}\rho_{\text{sep}})>0$ for all separable states. In the case of Bloch-diagonal states \eqref{rhobloch}, the witness evaluates to $\Tr(W_{\text{CCNR}}\rho_{AB})=1-\sum_z |\lambda_z|$, which is negative precisely when $\text{CCNR}(\rho_{AB})=\sum_z |\lambda_z|>1$.

Therefore, for the state \eqref{rhoxy} that violates the CCNR criterion, the corresponding entanglement witness is
\begin{equation}
\label{witxyccnr}
W^{(x,y)}_{\text{CCNR}} = \id_{16} - \sum_{z=1}^{16} \text{sgn}(\lambda_z) s^{(x,y)}_z A_z \otimes B_z.   
\end{equation}

Summing up the witnesses for each pair ($x,y$) and normalizing, we obtain
\begin{equation}
\label{witPM}
W_{\text{PM}}=\id_{16} - \frac{1}{4^3}\sum_{x,y,z} w_{xyz}2A_z\otimes 2B_z,    
\end{equation} 
where $w_{xyz}=(1/16)\text{sgn}(\lambda_z)s^{(x,y)}_z$, which corresponds to Eq.~\eqref{witcoeff}, and $C_z=2A_z\otimes 2B_z$ are Charlie's product observables. This expression exactly translates to the form of the prepare-and-measure witness~\eqref{wit} with $w_{xyz}$ given in~\eqref{witcoeff}. Then we have $\Tr(W_{\text{PM}}\rho_{AB})=1-\sum_z\left|\lambda_z\right|$. On the other hand, the separable bound 
$\Tr(W_{\text{PM}}\rho_{\text{sep}})\ge0$ follows directly, since each witness term $(x,y)$ in \eqref{witxyccnr} is non-negative for separable states. Thus, the prepare-and-measure witness in Eq. \eqref{witPM} perfectly functions as an entanglement witness when the observables are fixed as $C_z=4A_z\otimes B_z$ and the unitaries as $U_x = 2A_x$ and $V_y =2B_y$. In this case, the witness becomes equivalent to the CCNR criterion when applied to Bloch-diagonal states. 

Besides, in Proposition~\ref{prop4}, we also proved the separable bound of the correlation witness for the observables $C_z =4A_z\otimes B_z$ and showed that this bound can be saturated with specific separable states. Furthermore, using a see-saw optimization method, we have provided strong numerical evidence that this separable bound cannot be overcome even when allowing for arbitrary two-ququart observables beyond the fixed choice $C_z =4A_z\otimes B_z$. From this, the Main Result~\ref{theorem1} follows for Bloch-diagonal states.

\section{Examples of bound entangled states detected by the PM witness}
\label{examp}

In this section, we illustrate the power of our correlation witness based on four-dimensional quantum messages, using Main Result~\ref{theorem1}, on several bipartite entangled quantum systems. In Sec.~\ref{exampBD}, we present seven distinct families of $4\times 4$ entangled states in the Bloch-diagonal form of Eq.~(\ref{rhobloch}), among which four are PPT bound entangled. In Sec.~\ref{exampNBD}, we turn to PPT states with non-diagonal structure. Specifically, we provide examples of $4\times 4$ states, followed by $3\times 3$ states, embedded into the $4\times 4$-dimensional state space. Additionally, we present PPT bound entangled states in higher dimensions. Notably, we prove that for every dimension $D>4$, there exists a $D\times D$ bound entangled state whose entanglement is detected more robustly by our correlation inequality than by the violation of the CCNR criterion.  

\subsection{Bloch-diagonal examples}
\label{exampBD}

We consider distinct $4\times4$ entangled states $\rho_{\mathrm{AB}}^{(i)}$ for $i=1,\dots,7$, each admitting the Bloch-diagonal form \eqref{rhobloch}. The corresponding coefficients $\{\lambda_k\}$ are listed in Table~\ref{table1}, with the associated operators $A_k$ and $B_k$ specified in the table caption. Among these examples, four represent PPT bound entangled states.

Table~\ref{table2}, on the other hand, summarizes our numerical results for each state. Each row corresponds to a distinct state, for which we report the following entanglement parameters:
\begin{itemize}
  \item $\mathcal{N}$, the negativity of entanglement \cite{Vidal2002};
  \item the CCNR value, as given by the computable cross-norm (realignment) criterion for separability \cite{Rudolph2005,Chen2003}.
\end{itemize}

Note that for each state $\rho_{\mathrm{AB}}^{(i)}$ we have constructed a corresponding correlation witness with coefficients given in Eq.~\eqref{witcoeff}.  From Eq.~\eqref{Qrhoab} it then follows that
\begin{equation}
Q\bigl(\rho_{\mathrm{AB}}^{(i)}\bigr)
= \mathrm{CCNR}\bigl(\rho_{\mathrm{AB}}^{(i)}\bigr)\;\times\;Q_{\mathrm{sep}},    
\end{equation}
for $i=1,\dots,7$. Here we set $Q_{\mathrm{sep}}=4^3$, as established by the extensive see-saw optimization described in Sec.~\ref{Qsep}.  

\begin{table}[t!]
\centering
\begin{tabular}{c r r r r r r r} 
\toprule
No. & 1 & 2 & 3 & 4 & 5 & 6 & 7 \\
\midrule
$\lambda_k \backslash \text{State}$ & $\rho_{\text{ME}}$ & $\rho^{\text{W}}_{\text{as}}$ & $\rho^{\text{W}}_{\text{loc}}$ & $\rho^{\mathcal{F}}_{R6}$ & $\rho^{\mathcal{F}}_{R8}$ & $\rho_{\text{BPD}}$ & $\rho_{\text{Sentis}}$ \\
\midrule
$\lambda_1$  &  $1/4$  &  $1/4$  &  $1/4$ &  $1/4$  &  $1/4$   &  $1/4$  &   $1/4$ \\
$\lambda_2$  &  $1/4$  & $-1/12$ &  $-q$  &  $0$    &  $r_4$  &  $1/12$ &  $-s_1$ \\
$\lambda_3$  & $-1/4$  & $-1/12$ &  $-q$  &  $0$    & $-r_4$  & $-1/12$ &  $-s_1$ \\
$\lambda_4$  &  $1/4$  & $-1/12$ &  $-q$  & $-r_1$  &  $r_3$  & $-1/12$ &  $-s_1$ \\
$\lambda_5$  &  $1/4$  & $-1/12$ &  $-q$  &  $r_2$  &  $r_1$  & $-1/12$ &  $-s_1$ \\
$\lambda_6$  &  $1/4$  & $-1/12$ &  $-q$  &  $r_2$  &  $r_4$  &  $1/12$ &  $-s_1$ \\
$\lambda_7$  & $-1/4$  & $-1/12$ &  $-q$  &  $0$    & $-r_4$  &  $1/12$ &   $s_2$ \\
$\lambda_8$  &  $1/4$  & $-1/12$ &  $-q$  & $-r_2$  & $-r_1$  &  $1/12$ &   $s_3$ \\
$\lambda_9$  & $-1/4$  & $-1/12$ &  $-q$  &  $0$    & $0$     &  $1/12$ &  $-s_1$ \\
$\lambda_{10}$ & $-1/4$ & $-1/12$ &  $-q$  &  $0$    & $-r_4$  &  $1/12$ &   $s_3$ \\
$\lambda_{11}$ &  $1/4$ & $-1/12$ &  $-q$  &  $r_2$  & $-r_4$  &  $1/12$ &   $s_2$ \\
$\lambda_{12}$ & $-1/4$ & $-1/12$ &  $-q$  &  $0$    & $0$     & $-1/12$ &  $-s_1$ \\
$\lambda_{13}$ &  $1/4$ & $-1/12$ &  $-q$  &  $r_3$  & $0$     &  $1/12$ &   $s_2$ \\
$\lambda_{14}$ &  $1/4$ & $-1/12$ &  $-q$  &  $0$    & $-r_4$  &  $1/12$ &  $-s_1$ \\
$\lambda_{15}$ & $-1/4$ & $-1/12$ &  $-q$  &  $0$    & $-r_4$  & $-1/12$ &   $s_3$ \\
$\lambda_{16}$ &  $1/4$ & $-1/12$ &  $-q$  &  $r_1$  & $0$     &  $1/12$ &  $-s_1$ \\
\bottomrule
\end{tabular}
\caption{\label{table1} 
\textbf{Coefficients of the Bloch-diagonal $4\times 4$ states.} The table lists the coefficients $\lambda_k$ ($k=1,\ldots,16$) for the explicit decomposition of the entangled quantum states under study. The parameters are defined as: $q=19/340$, $r_1=(\sqrt{2}-1)/4$, $r_2=(2-\sqrt{2})/4$, $r_3=r_2-r_1$, $r_4=r_2/2$, and $s_1=0.0557066$, $s_2=0.0142664$, $s_3=0.0971467$. Each state $\rho_{AB}$ is defined in Eq.~\eqref{rhobloch} using the product of Pauli operators $A_k=B_k=(1/2)\sigma_{k_0}\otimes\sigma_{k_1}$ for $k=1,\ldots,16$, except for $\rho^{\mathcal{F}}_{R6}$ and $\rho^{\mathcal{F}}_{R8}$. For $\rho^{\mathcal{F}}_{R6}$, $B_6$ and $B_{11}$ are swapped with $A_{11}$ and $A_6$, respectively. For $\rho^{\mathcal{F}}_{R8}$, the swaps are $B_{10} \leftrightarrow A_{11}$, $B_{11} \leftrightarrow A_{10}$, and $B_{14} \leftrightarrow A_{15}$, $B_{15} \leftrightarrow A_{14}$.
}
\end{table}

To characterize the state, we introduce a visibility parameter $v$ by mixing $\rho_{\mathrm{AB}}^{(i)}$ with isotropic (white) noise:
\begin{equation}
\label{statemixing}
\rho_{\mathrm{AB}}^{(i)}(v)
= v\,\rho_{\mathrm{AB}}^{(i)} \;+\;(1 - v)\,\frac{\id_{16}}{16}\,,
\end{equation}
where $0\le v\le1$. Here we introduce several visibility thresholds for the mixed state $\rho_{\mathrm{AB}}^{(i)}(v)$ from Eq.~\eqref{statemixing}:
\begin{itemize}
  \item \emph{Critical visibility} $v_{\mathrm{PM}}$: the smallest $v$ for which our prepare-and-measure witness still detects entanglement. For Bloch-diagonal states, the value satisfies $\mathrm{CCNR}[\rho_{\mathrm{AB}}^{(i)}(v_{\mathrm{PM}})]=1$, according to Main Result~\ref{theorem1}. Noting that the maximally mixed state $\id_{16}/16$ has a single nonzero coefficient $\lambda_1=1/4$, one finds
  \begin{equation}
   v_{\mathrm{PM}}
    \;=\;\frac{3}{4\,\mathrm{CCNR}(\rho_{\mathrm{AB}}^{(i)}) - 1}\,;  
  \end{equation}
   
   \item \emph{Metrological threshold} $v_{\mathrm{metro}}$: the minimum visibility above which $\rho_{\mathrm{AB}}^{(i)}(v)$ outperforms all separable states in a quantum metrology task~\cite{Toth2020activating, Horodecki2021};

  \item \emph{LHV threshold} $v_{\mathrm{loc}}^{(\mathrm{proj})}$: a lower bound on the visibility below which $\rho_{\mathrm{AB}}^{(i)}(v)$ admits a local hidden-variable (LHV) model for projective measurements~\cite{Brunner2014Bell};

  \item \emph{Separability threshold} $v_{\mathrm{sep}}$: an (often tight) upper bound on the visibility below which the state becomes separable~\cite{Guhne2009,Horodecki2009}.
\end{itemize}

Before presenting the states case-by-case, we provide details on the calculation of $v_{\text{metro}}$. Using the numerical method introduced in Ref.~\cite{Toth2020activating}, we find that the optimal Hamiltonian takes the form 
\begin{equation}
H_{AB} = H_A\otimes \id_4 + \id_4 \otimes H_B,
\label{HDD}
\end{equation}
with $H_A = H_B = H$. For all states except $\rho_{R8}^{\mathcal{F}}$, the matrix is taken as $H = \id_2\times Z=\text{diag}(+1,-1,+1,-1)$. For $\rho_{R8}^{\mathcal{F}}$, the optimal choice is $H = Z\times \id_2=\text{diag}(+1,+1,-1,-1)$. In both cases, $H$ is a local, traceless observable with eigenvalues $\pm 1$, so that the separability limit for the quantum Fisher information is $\mathcal{F}_{\text{sep}}=8$~\cite{Toth2018}. Consequently, $v_{\text{metro}}$ is defined as the critical visibility in the mixture~(\ref{statemixing}) for which the maximal quantum Fisher information $\mathcal{F}_{Q}^{\text{max}}(\rho_{AB}^{(i)}(v))$ reaches 8. 

\begin{table}[t!]
\centering
\begin{tabular}{c c r r c c c c}
\toprule
No. & State & $\mathcal{N}$ & CCNR & $v_{\text{PM}}$ & $v_{\text{metro}}$ & $v^{(\text{proj})}_{\text{loc}}$ & $v_{\text{sep}}$ \\
\midrule
1 & $\rho_{\text{ME}}$ & $1.5$ & $4$ & $0.2000$ & $0.5509$ & $0.3611$ & $0.2000$ \\
2 & $\rho_{\text{as}}^{\text{W}}$ & $0.2500$ & $1.5$ & $0.6000$ & $0.8496$ & $0.7500$ & $0.2000$ \\
3 & $\rho^{\text{W}}_{\text{loc}}$ & $0.1471$ & $1.0882$ & $0.8947$ & - & $1.0000$ & $0.2983$ \\
4 & $\rho^{\mathcal{F}}_{R6}$ & $0$ & $1.0858$ & $0.8974$ & $0.9183$ & - & $0.7446$ \\
5 & $\rho^{\mathcal{F}}_{R8}$ & $0$ & $1.0858$ & $0.8974$ & $0.9183$ & - & $0.7446$ \\
6 & $\rho_{\text{BPD}}$ & $0$ & $1.5$ & $0.6000$ & - & - & $0.6000$ \\
7 & $\rho_{\text{Sentis}}$ & $0$ & $1.0856$ & $0.8976$ & - & - & $0.7814$ \\
\bottomrule
\end{tabular}
\caption{\label{table2} 
\textbf{Entanglement and critical noise parameters of Bloch-diagonal $4\times 4$ states.} Each state $\rho_{AB}^{(i)}$ ($i=1,\ldots,7$) is shown in a separate row and expressed in the Bloch decomposition~\eqref{rhobloch} with coefficients $\lambda_k$ and operators $A_k$, $B_k$ from Table~\ref{table1}. Columns $\mathcal{N}$ and CCNR report the entanglement negativity~\cite{Vidal2002} and the CCNR value~\cite{Rudolph2005,Chen2003}, which detect entanglement when $\mathcal{N}>0$ and $\text{CCNR}>1$. For $4\times 4$ systems, the maximum possible values are $\mathcal{N}=1.5$ and $\text{CCNR}=4$. Remaining columns list critical visibility parameters as defined in the main text.}
\end{table}

\smallskip
\noindent Detailed descriptions of all seven quantum states follow.
\smallskip

\noindent\textbf{1. Maximally entangled state:} The state $\rho_{\text{ME}} = \ketbra{\Phi_4^+}{\Phi_4^+}$ with $ \ket{\Phi_4^+}= 1/2\sum_{k=1}^4\ket{k}\ket{k}$ has Bloch-decomposition coefficients as given in the first column of Table~\ref{table1}. Its CCNR value is $\sum_{k=1}^{16} |\lambda_k| = 4$. When mixed with white noise,
\begin{equation}
\rho_{\text{ME}}(v) = v\ketbra{\Phi_4^+}{\Phi_4^+} + (1-v)\id_{16}/16,
\end{equation}
the resulting isotropic state is entangled for $v>0.2$, i.e., $v_{\text{sep}}=0.2$ (as determined by the PPT criterion~\cite{Peres1996}), and $v_{\text{PM}}=0.2$ coincides with this value. Moreover, there exists a local model for projective measurements for $v \le \sum_{k=2,3,4}1/(3k) \simeq 0.3611$~\cite{Almeida2007}. On the other hand, the state is useful for metrology with the Hamiltonian~(\ref{HDD}) if $v> (1/32)\times(7 + \sqrt{113}) \simeq  0.5509$ (see the Supp. Mat. in Ref.~\cite{Toth2020activating}).
\\[6pt]
\noindent\textbf{2. Antisymmetric $4\times 4$ Werner state:}  This state belongs to the $d\times d$ Werner family~\cite{Werner1989}
\begin{equation}
\rho_{\text{W}}(p) = p P_{\text{as}}/6 + (1-p)P_{\text{sym}}/10, 
\label{Wernerstate}
\end{equation}
where $P_{\text{as}} = (\id_4 - V)/2$ and $P_{\text{sym}} = (\id_4 + V)/2$ are the projectors on the respective antisymmetric and symmetric spaces, and $V = \sum_{ij} \ketbra{ji}{ij}$ is the swap operator. The state is defined as $\rho^{\text{W}}_{\text{as}}=P_{\text{as}}/6$ (corresponding to $p=1$), with a CCNR value of $\sum_{k=1}^{16} |\lambda_k| = 3/2$. When mixed with isotropic noise of weight $(1-v)$, the state remains within the Werner class and is known to be separable only when it is PPT, hence, it is entangled for $v>0.2$. It admits a local hidden variable model for projective measurements when $v\le 3/4$ (Ref.~\cite{Werner1989}) and it becomes metrologically useful when $v>0.8496$. Note, however, that $v_{\text{PM}}=0.6$, which is much higher than $v_{\text{sep}}=0.2$. These states exhibit a high degree of symmetry and shareability~\cite{Lancien2016}, and can also be used for quantum data hiding~\cite{DiVincenzo2002}.  
\\[6pt]
\noindent\textbf{3. $4\times 4$ Werner state with a local hidden variable model:} 
The state is defined as $\rho_{\text{loc}}^{\text{W}} = \rho_{\text{W}}(27/34)$ in Eq.~(\ref{Wernerstate}). The parameter $p=27/34$ marks the threshold below which Werner's local hidden variable model applies~\cite{Werner1989}. Although the mixture~(\ref{statemixing}) is entangled for $v>0.2983$ (by the PPT criterion~\cite{Peres1996}) and $v_{\text{PM}}=0.8947$, the state is not useful for metrology since its maximum quantum Fisher information (with Hamiltonian~(\ref{HDD})) is only $\mathcal{F}^{\text{max}}_Q = 5.2271<8$. 
\\[6pt]
\noindent\textbf{4. A metrologically useful rank-6 bound entangled state:}  Denoted by $\rho_{R6}^{\mathcal{F}}$, this state is given in Supp. Mat. of Ref.~\cite{Toth2018} and is extremal within the set of $4\times 4$ PPT entangled states~\cite{Badziag2014}. Ref.~\cite{Pal2021} shows that it is local unitarily equivalent to the state defined in Ref.~\cite{Horodecki2005}, which is part of a family of $2d\times 2d$ states with $d\ge 2$ known as private states. The separability threshold of the mixture~(\ref{statemixing}) is $v_{\text{sep}}=0.7446$ (determined by not having a two-copy PPT symmetric extension~\cite{Doherty2002}, implemented using the QETLAB package~\cite{qetlab}). On the other hand, the critical visibility of the state detected by our PM witness is $v_{\text{PM}}=0.8974$. The state has $\mathcal{F}^{\text{max}}_Q = 32 - 16\sqrt{2} \simeq 9.3726>8$ with the optimal Hamiltonian~(\ref{HDD}), where $H = \text{diag}(+1,+1,-1,-1)$, making it useful for metrology for $v>0.9183$. No local hidden variable model is known for this state. 
\\[6pt]
\noindent\textbf{5. A metrologically useful rank-8 bound entangled state:} The state $\rho_{R8}^{\mathcal{F}}$ has the same metrological performance as $\rho_{R6}^{\mathcal{F}}$ but differs in rank and is not invariant under partial transposition. An equivalent state up to local unitary transformations is given in Ref.~\cite{Pal2021}. The state is part of a $2d\times 2d$ family of states described in Section V of Ref.~\cite{Pal2021} and is also available in the QUBIT4MATLAB package~\cite{qubit4matlab} (see the routine \verb|BES_metro.m|). It has the same separability threshold $v_{\text{sep}}=0.7446$, PM threshold $v_{\text{PM}}=0.8974$ and maximum quantum Fisher information $\mathcal{F}_Q^{\text{max}} = 32 - 16\sqrt{2} \simeq 9.3726$ as for the state $\rho_{R6}^{\mathcal{F}}$. In this case, the optimal Hamiltonian~(\ref{HDD}) is $H = \text{diag}(+1,-1,+1,-1)$, and it is useful for metrology when $v>0.9183$. Again, we are not aware of any local hidden variable model for this state. 
\\[6pt]
\noindent\textbf{6. A $4\times 4$ Bloch-product-diagonal state:} Introduced in Ref.~\cite{Benatti2004}, the state $\rho_{\text{BPD}}$ when mixed with white noise is determined to be entangled (via the CCNR criterion~\cite{Rudolph2005,Chen2003}) for $v>0.6$, corresponding to $v_{\text{sep}}=v_{\text{PM}}=0.6$. The routine \verb|BES_CCNR4x4.m| in the QUBIT4MATLAB package~\cite{qubit4matlab} defines this state. Some other references on $\rho_{\text{BPD}}$ include Refs.~\cite{Moroder2012calibration,Lukacs2022,deGois2023}. Numerical analysis indicates that this state maximizes the CCNR value (at 1.5) among $4\times 4$ PPT entangled states~\cite{Lukacs2022}. A very similar PPT state is provided in Ref.~\cite{Moerland2024} with the same CCNR value. However, we are not aware of any local hidden variable model, and the state is not metrologically useful since $\mathcal{F}_Q^{\text{max}} = 5.3333<8$ for any Hamiltonian of the form~(\ref{HDD}).
\\[6pt]
\noindent\textbf{7. Sentis et al.~state:} The state $\rho_{\text{Sentis}}$ is defined in Ref.~\cite{Sentis2018} which detects bound entanglement with high statistical significance, making it robust for experimental certification. When mixed with white noise, it is determined to be entangled for $v>0.7814$ using the Breuer-Hall positive maps~\cite{Breuer2006, Hall2006}. It is not useful metrologically, as $\mathcal{F}_Q^{\text{max}} = 4.8913<8$ with any Hamiltonian of the form~(\ref{HDD}), and is not known to admit a local hidden variable model. The critical visibility of the mixture~(\ref{statemixing}) corresponding to both the CCNR criterion and to the prepare-and-measure setup is $v_{\text{PM}}=0.8976$.

\subsection{Non-Bloch-diagonal examples}
\label{exampNBD}

According to Main Result~\ref{theorem1}, the PM witness~\eqref{wit}
with coefficients~\eqref{witcoeff} detects entanglement of a $4\times 4$ state $\rho$ whenever the trace criterion
\begin{equation}
\label{Srho}
S(\rho) \;=\;\sum_{k=1}^{16}\bigl|t_{kk}\bigr| \;\le\;1
\end{equation}
is violated, where $t_{kk} \;=\;\Tr\bigl(\rho\,A_k\otimes B_k\bigr)$ are the diagonal entries of the correlation matrix $T=(t_{k\ell})$.

Although for non-Bloch-diagonal states this trace criterion is strictly weaker than the CCNR criterion (cf. Eq.~\eqref{JvN}), it still catches weakly entangled states. In particular, we demonstrate that it can certify PPT entanglement not only in $4\times4$ systems but even in $3\times3$ states with a non-diagonal Bloch structure.

Here we present three illustrative examples: In Sec. \ref{E4}, starting from the $\rho_{\text{BPD}}$ state (i.e., state No.~6 in Table~\ref{table1}), we mix in asymmetric noise to obtain a modified $4\times4$ state. In Sec.~\ref{E3}, through a brute force search combined with a semidefinite programming (SDP) algorithm~\cite{Vandenberghe1996}, we construct a $3\times3$ bound entangled state embedded into the $4\times4$ space. In Sec.~\ref{Ege4}, we begin with a known $4\times4$ Bloch-diagonal PPT state and add suitably chosen noise in a larger Hilbert space of dimension $D\times D$ for $D\ge4$. For every $D>4$, this procedure produces a PPT state whose entanglement is certified by our prepare-and-measure witness based on four-dimensional messages, whereas the CCNR criterion fails to detect it.

\subsubsection{Example for \texorpdfstring{$D = 4$}{D = 4}}
\label{E4}
Specifically, we consider the $4\times 4$ state 
\begin{equation}
\rho_{\text{asym}}(v) = v\rho_{\text{BPD}} + (1-v)\id_4\otimes\ketbra{00}{00}.   
\end{equation}
This state may be regarded as a noisy version of the Bloch-product-diagonal state $\rho_{\text{BPD}}$, where the added noise gives rise to asymmetric marginal terms in the correlation tensor. Consequently, for $v<1$, it can no longer be expressed in Bloch-diagonal form. A closed form expression for $S(\rho_{\text{asym}}(v))$ is given by 
\begin{equation}
S\left(\rho_{\text{asym}}(v)\right)=\sum_{k=1}^{16}|t_{kk}(v)|=\frac{3v}{2},
\end{equation}
and hence the threshold $S=1$ for the trace criterion \eqref{Srho} is reached at $v_{\text{PM}}=0.6$. On the other hand, at this same noise level, the CCNR value becomes 
\begin{equation}
\text{CCNR}(\rho_{\text{asym}}(v_{\text{PM}}))=\frac{7}{10}+\frac{\sqrt 3}{5}\simeq 1.0464>1.
\end{equation}
So the CCNR criterion, as expected, is more powerful than the trace criterion for detecting entanglement.

\subsubsection{Example for \texorpdfstring{$D = 3$}{D = 3}}
\label{E3}
We now embed a $3\times 3$ PPT state $\rho_{3\times3}$ into the $4\times 4$ state space and seek the maximal value of $S(\rho_{3\times 3})=\sum_{k=1}^{16}|t_{kk}|$. By Main Result~\ref{theorem1}, the prepare-and-measure witness detects entanglement as soon as $S>1$. A brute-force search shows that the optimal embedded state achieves $S(\rho_{3\times 3}) = \frac{5}{2} - \sqrt{2}\;\simeq\;1.0858$. The algorithm below generates this state by iteratively calling a single-shot semidefinite program to maximize $S(\rho) \;=\;\sum_{k=1}^{16} \bigl|t_{kk}\bigr|$ over all $3\times 3$ PPT states in the criterion \eqref{Srho}:
\smallskip

\noindent\textbf{A brute force SDP algorithm.} 

\begin{enumerate}
  \item Pick a sign vector $\vec s\in\{\pm 1\}^{16}$. We loop over all 16-element sign vectors $\vec s$ with $s_1= + 1$.  
  
  \item At each iteration, for the current sign vector $\vec s$, solve the following SDP:
  \begin{equation}\label{singleshotSDP}
	\begin{aligned}
		\max_{\{t_{kk}\}} \quad & \sum_{k=1}^{16}s_k t_{kk}\\
		\st \quad & \rho = \sum_{k\ell}t_{k\ell}A_k\otimes A_\ell,\\\quad & \rho = (F_A\otimes F_B) \rho (F_A\otimes F_B),\\
        \quad & \rho\succeq 0,\\
        \quad & \mathrm{Tr}(\rho)=1,\\
        \quad & \mathrm{PT}(\rho)\succeq 0.\\
    	\end{aligned}
\end{equation}
Here, the 3-dimensional embedding into the $4\times4$ space is enforced by the projectors $F_A=F_B=\mathrm{diag}(1,1,1,0)$, while the state $\rho$ is imposed to be positive, normalized, and PPT. Under these constraints, we solve a single-shot SDP to maximize the linear objective $S=\sum_{k=1}^{16} s_k t_{kk}$.

\item Track and update the best objective value $S$ and corresponding variables $\vec s$ and $t_{k\ell}$.

\end{enumerate}

As a result, we obtain the PPT state $\rho_{3\times3}$ that maximizes $S$ in \eqref{Srho} over all sign vectors $\vec s$. We give the state by the coefficients $t_{k\ell}$ of its correlation matrix $T$ in the decomposition $\rho_{3\times 3}=\sum_{k\ell}t_{k\ell}A_k\otimes A_\ell$. First, the diagonal entries $t_{kk}$ of $T$ read
\begin{equation}
\left(\frac{1}{4},r_4,-r_4,-r_3,0,r_4,r_4,0,0,r_4,r_4,0,r_1,r_4,-r_4,r_1\right).
\end{equation}
Next, the strictly upper-triangular nonzero off-diagonal elements are
\begin{equation}
\begin{aligned}
t_{1,13} &= \tfrac18,   &\quad t_{1,16} &= -\tfrac18,\\
t_{7,10} &= -r_4,         &\quad t_{6,11} &= r_4,\\
t_{2,14} &= r_4,        &\quad t_{3,15} &= -r_4,\\
t_{4,13} &= \tfrac{r_3}{2}, &\quad t_{4,16} &= -\tfrac{r_3}{2}.
\end{aligned}
\end{equation}
Since the correlation matrix $T$ is symmetric ($t_{k\ell}=t_{\ell k}$) in our case, the corresponding lower-triangular entries follow immediately. Here the parameters $r_i$ (in particular $r_3$ and $r_4$) are defined in the caption of Table~\ref{table1}. From the above specification, one finds 
\begin{align*}
\text{CCNR}(\rho_{3\times 3}) &\;=\;\|T\|_1 \;\simeq\;1.1163,\\
S(\rho_{3\times 3})&\;=\;\sum_{k=1}^{16} |t_{kk}|\;=\;\frac{5}{2}-\sqrt2\;\simeq\;1.0858.     
\end{align*}
Hence, $1 < S(\rho_{3\times 3}) < \text{CCNR}(\rho_{3\times 3})$, and both criteria successfully detect the entanglement of this $3\times 3$ PPT state.

\subsubsection{Example for \texorpdfstring{$D \ge 4$}{D ≥ 4}}
\label{Ege4}

We now demonstrate that our correlation witness can detect even higher than $4\times 4$ dimensional bound entangled states. In fact, whenever $Q(\rho_{AB})>Q_{\text{sep}}$, where $Q_{\text{sep}}=4^3$ in Eq.~(\ref{wit}), the state is detected entangled. For instance, this allows us to detect PPT bound entangled states in arbitrary high dimensions. To this end, let us pick a $4\times 4$ Bloch-diagonal PPT bound entangled state $\rho$ detected by the CCNR criterion, $\text{CCNR}(\rho)>1$: that is, we have $Q(\rho)>Q_{\text{sep}}$. Let us add a certain amount of ($1-v$) isotropic noise with $D\ge4$ to the $4\times 4$ state to obtain 
\begin{equation}
\rho_{D,v}=v\rho + (1-v)\frac{\id_{D^2}}{D^2},  
\label{rhoDAB}
\end{equation}
which remains PPT. Now let the message encoding maps of Alice and Bob be a fixed CPTP map followed by the unitary $U_x\otimes V_y$. In particular, let Alice and Bob apply a local CPTP map which projects the $4\times 4$-dimensional space where $\rho$ is living intact, and whenever the state lies outside that subspace, re-prepares the $4\times 4$ product state $\rho_A\otimes\rho_B$ as follows
\begin{equation}
\varepsilon(\rho_{D,v})
= P\,\rho_{D,v}\,P
\;+\;\Tr\bigl[(\id_{D^2} - P)\,\rho_{D,v}\bigr]\;\rho_A\otimes\rho_B,  
\label{epsrho}
\end{equation}
where $P$ is a projector $P: \mathbb{R}^{D\times D} \rightarrow \mathbb{R}^{4\times 4}$. Then we have 
\begin{equation}
\begin{aligned}
P\rho_{D,v}P&=v\rho + (1-v)\frac{\id_{4^2}}{D^2},\\
\Tr\bigl[(\id_{D^2} - P)\,\rho_{D_v}\bigr]&=(1-v)\left(1-\frac{4^2}{D^2}\right),
\end{aligned}
\label{stateproj}
\end{equation}
which results in the $4\times 4$ state
\begin{equation}
\varepsilon(\rho_{D,v})
= v\rho + (1-v)\frac{\id_{4^2}}{D^2}+(1-v)\left(1-\frac{4^2}{D^2}\right)\rho_A\otimes\rho_B.
\label{stateprojfull}
\end{equation}
Let us now consider two situations depending on the form of the re-prepared product states.
\\[6pt]
\noindent\textbf{Case of isotropic noise.} Let us have $\rho_A=\rho_B=\id_4/4$. In that case $\varepsilon(\rho_{D,v})$ becomes
\begin{equation}
\varepsilon(\rho_{D,v})
= v\rho + (1-v)\frac{\id_{4^2}}{4^2}.
\label{stateprojiso}
\end{equation}
Note that the state after the CPTP map does not depend on $D$. We then apply the local unitaries $2A_x$ and $2B_y$ on the state \eqref{stateprojiso}, and then measure the observable $C_z=4A_z\otimes B_z$. Since the resulting state is Bloch-product-diagonal, our correlation witness detects entanglement according to Main Result~\ref{theorem1} as soon as
\begin{equation}
\text{CCNR}(\epsilon(\rho_{D,v})) = v\,\mathrm{CCNR}(\rho) + \frac{(1-v)}{4}
> 1,
\label{ccnr_iso}
\end{equation}
which rearranges to
\begin{equation}
v_{\text{PM}} \;=\;
\frac{3}{4\,\mathrm{CCNR}(\rho) - 1}
\end{equation}
for the critical visibility defined by $\text{CCNR}(\epsilon(\rho_{D,v_{\text{PM}}}))=1$. This formula tells us that the critical visibility $v_{\text{PM}}$ is strictly less than unity whenever $\text{CCNR}(\rho)>1$ for all $D\ge 4$. So, by applying the map~(\ref{epsrho}) with $\rho_A=\rho_B=\id_4/4$, one can achieve a sizable violation of the correlation inequality for high-dimensional PPT states~(\ref{rhoDAB}). 

Also note that the CCNR value of the state~(\ref{rhoDAB}) is 
\begin{equation}
\label{ccnrvD}
\text{CCNR}(\rho_{D,v})=v\text{CCNR}(\rho)+\frac{1-v}{D},     
\end{equation}
which is strictly smaller than the expression in \eqref{ccnr_iso} for $D>4$. So, for all $D>4$, when applied to the family of states~\eqref{rhoDAB}, our prepare-and-measure witness strictly outperforms the CCNR criterion.  In fact, for each $D>4$ there exists a PPT state within the family \eqref{rhoDAB} whose entanglement is revealed by our correlation witness yet goes undetected by the CCNR test.
\\[6pt]
\noindent\textbf{Case of product noise.} Let us choose the state $\rho_{\text{BPD}}$ in Eq. \eqref{rhoDAB}, for which $\text{CCNR}(\rho_{\text{BPD}})=3/2$, and set $\rho_A=\rho_B=\ketbra{0}{0}$ in the map~(\ref{epsrho}). With these choices, the witness value $S(v,D)$ in \eqref{Srho} evaluates to 
\begin{equation}
S(v,D)=\frac{2v}{3}+\frac{3}{4}-\frac{8(1-v)}{D^2}+\left|\frac{v}{3}-\frac{1}{4}+\frac{4(1-v)}{D^2}\right|.    
\end{equation}
Then, for $D\to\infty$ we have
\begin{equation}
S(v,\infty) =
\begin{cases}
\frac{3v + 2}{3}, & \text{if } v < \frac{4}{3} \\
\frac{v + 2}{3}, & \text{if } v \geq \frac{4}{3}.
\end{cases}    
\end{equation}
Hence, for $\rho_{\mathrm{BPD}}$ in the family \eqref{rhoDAB}, our prepare-and-measure witness detects entanglement for any visibility $v$ satisfying $S(v,\infty)>1$, which here holds for all $v>0$. In contrast, the CCNR criterion \eqref{ccnrvD} only detects entanglement in the limit $D\to\infty$ for $v>2/3$.  Thus, for this particular high-dimensional state, our correlation witness is significantly more sensitive than the CCNR test.

\section{Proposal and feasibility of experimental implementations}

In recent years, impressive advances have been achieved in the preparation, manipulation, and observation of high-dimensional quantum systems~\cite{Friis2018}.

In addition, the distribution of high-dimensional entanglement over long distances has become available, for example via multicore fibers~\cite{Hu2020, Hu2022} and has been complemented by pioneering demonstrations of spectrally indistinguishable photons across long distances~\cite{Zhan2025}. Quantum computers based on higher-dimensional systems, i.e. qudits rather than qubits, are also coming into sight~\cite{Lanyon2008}. Partly motivated by these advances, substantial efforts have recently been made to prepare and observe bound entangled states, initially in the multipartite setting and, more recently, in the bipartite setting.

Multipartite bound entangled states have been realized in several platforms: a four-qubit Smolin state~\cite{Smolin2001,Augusiak2006} was prepared both with polarization entangled photons~\cite{Amselem2009,Lavoie2010,Dobek2013} and a string of trapped ions~\cite{Barreiro2010}, and a three-qubit system was generated in a liquid-state NMR system~\cite{Kampermann2010}. Importantly, in the multipartite setting above, entanglement may still be distilled if two of the parties work together. Later the fundamentally different bipartite case was explored experimentally, though to a lesser extent: DiGuglielmo et al.~\cite{DiGuglielmo2011} produced bound entangled Gaussian states using multimode entangled light, and Hiesmayr et al.~\cite{Hiesmayr2013} witnessed two-qutrit bound entangled states in their orbital angular momentum.

In the following, we propose to generate the Bloch-product-diagonal bound entangled state $\rho_{\text{BPD}}$ (state No.~6 in Table~\ref{table1}), which, according to numerical evidence, exhibits the highest CCNR value among $4\times 4$ PPT states, and to witness its correlations in our prepare-and-measure setup using polarization-entangled photons. 

We note that entanglement-witness setups~\cite{Amselem2009,Dobek2013} with similar photonic implementation have previously been used to detect multipartite bound entanglement in the noisy four-qubit Smolin state~\cite{Smolin2001}.

Let us first define the four maximally entangled Bell pairs:
\begin{align}
  \ket{\Psi_{1,2}} &= \ket{\Psi^{\pm}}
    = \frac{\ket{01}\pm\ket{10}}{\sqrt{2}}, \\[0.5ex]
  \ket{\Psi_{3,4}} &= \ket{\Phi^{\pm}}
    = \frac{\ket{00}\pm\ket{11}}{\sqrt{2}}.
\end{align}

The bound entangled state $\rho_{\mathrm{BPD}}$ is then given by
\begin{align}
  \rho_{\mathrm{BPD}}
  =& \sum_{i=1}^3 \frac{1}{6}\,\ket{\Psi_i}\bra{\Psi_i}_{AB}
       \otimes \ket{\Psi_i}\bra{\Psi_i}_{A'B'}\nonumber\\
    & + \sum_{i=1}^3 \frac{1}{6}\,\ket{\Psi_4}\bra{\Psi_4}_{AB}
       \otimes \ket{\Psi_i}\bra{\Psi_i}_{A'B'}\,, 
  \label{state_bellprod}
\end{align}
where $(AA')$ and $(BB')$ denote the ququart spaces of Alice and Bob, respectively. To prepare this state, we start from the double Bell-product
\begin{equation}
  \ket{\Psi_1}_{AB} \;\otimes\; \ket{\Psi_1}_{A'B'} 
  \label{state_psi1double}
\end{equation}
and with equal probability apply one of the following two procedures:

\begin{enumerate}
  \item With probability $1/2$: choose $k\in\{0,1,3\}$ uniformly at random, and apply the Pauli unitary $\sigma_k$ on both $\ket{\Psi_1}_{AB}$ and $\ket{\Psi_1}_{A'B'}$ to obtain
    \begin{equation}
      (\sigma_k \otimes \id)\,\ket{\Psi_1}_{AB}
      \;\otimes\;
      (\sigma_k \otimes \id)\,\ket{\Psi_1}_{A'B'}\,,   
    \end{equation}
    where $\sigma_0 = \id$.
  
  \item With probability $1/2$: first rotate $\ket{\Psi_1}_{AB}$ to $\ket{\Psi_4}_{AB}$ via $\sigma_2$. Then choose $k\in\{0,1,3\}$ uniformly at random and apply the unitary $\sigma_k$ to $\ket{\Psi_1}_{A'B'}$, yielding
    \begin{equation}
         \ket{\Psi_4}_{AB}
      \;\otimes\;
      (\sigma_k \otimes I)\,\ket{\Psi_1}_{A'B'}\,. 
    \end{equation}
\end{enumerate}
This procedure yields the state in Eq.~\eqref{state_bellprod}.

In our proposed experiment for generating $\rho_{\text{BPD}}$, we take an approach similar to Refs.~\cite{Amselem2009,Lavoie2010}. These works use spontaneous parametric down conversion to produce the two-photon Bell-product state~(\ref{state_psi1double}) in four distinct spatial modes driven with a single laser pump. We then apply the Pauli unitaries $\sigma_k$ via single-qubit bit-flip and phase-shift gates, and use a random-number generator to select among them as described above. As a result, this experimental procedure generates the bound entangled state given in Eq.~(\ref{state_bellprod}). 

On the other hand, the inputs $x$, $y$ and $z$ for our correlation witness are drawn from a random-number generator. The product Pauli unitaries $U_x=\sigma_{x_0}\otimes\sigma_{x_1}$ and $V_y=\sigma_{y_0}\otimes\sigma_{y_1}$ are implemented on the respective $AA'$ and $BB'$ shares of $\rho_{\text{BPD}}$ using single-qubit bit-flip and phase-shift gates. The outgoing photons are then coupled into single-mode fibers and sent to local polarization analyzers. There, avalanche photodiodes register four-fold coincidences, allowing us to measure the observables $C_z=\sigma_{z_0}\otimes\sigma_{z_1}\otimes\sigma_{z_0}\otimes\sigma_{z_1}$ on the full $AA'BB'$ system. 

To test for bound entanglement we can proceed in two steps. 
(i) We directly measure the correlation witness from Eqs.~(\ref{wit},\ref{witcoeff}). For each input pair $(x,y)$, we perform 16 measurements, yielding a total of $16^3$ measurements for all choices $x,y,z$. If the resulting witness exceeds the separable bound $Q_{\text{sep}}=4^3$, the prepared state is certified as entangled. (ii) The density matrix is reconstructed by full state tomography, which requires $16^2$ different measurement settings. We then check whether the reconstructed state is positive under partial transpose.

However, because $\rho_{\text{BPD}}$ has a non-full-rank partial transpose, verifying the PPT condition is highly sensitive to experimental noise (see Ref.~\cite{Lavoie2010NP}). To reach PPT one can add white noise to the state $\rho_{\text{BPD}}$ as has been done in Refs.~\cite{Lavoie2010, Dobek2013}. By choosing the noise weight appropriately, one can simultaneously (i) violate the correlation witness and (ii) enforce the PPT property, thus providing a robust experimental demonstration of bound entanglement within the prepare-and-measure setup.

Given the rapid experimental progress in photonic entanglement, we expect that a polarization-photon experiment to detect bound entanglement in the PM setup will soon be within experimental reach. Moreover, any other $4\times 4$ Bloch-diagonal state violating the CCNR criterion, such as the metrologically useful $\rho_{R6}^{\mathcal{F}}$ or $\rho_{R8}^{\mathcal{F}}$ from Table~\ref{table1}, can be generated and observed by a similar procedure: randomizing Bell-product pairs and then applying our correlation test. 

Furthermore, other experimental platforms, both photonic and non-photonic, also hold great promise~\cite{Friis2018}. In photonic systems high-dimensionally entangled states can be produced in the temporal domain, the frequency domain, via orbital angular momentum entanglement, or by combining these degrees of freedom to create hyperentangled states~\cite{Barreiro2005,Wang2015a}. Non-photonic platforms likewise may provide routes to complex high-dimensional entanglement (including bound entanglement), for example with Cesium atoms~\cite{Anderson2015}, transmon superconducting qubits~\cite{Kumar2016}, or nitrogen-vacancy centers~\cite{Wang2015b}.

\begin{table}[t]
  \centering
  \sisetup{round-mode=places, round-precision=5}
  \begin{tabular}{c S[table-format=1.5] S[table-format=1.5]}
    \toprule
    {$d$} & {CCNR (PM witness)} & {CCNR (Bell-diagonal)} \\
    \midrule
    3 & 1.16646 & 1.18585 \\
    5 & 1.31934 & 1.41464 \\
    7 & 1.43684 & 1.52053 \\
    \bottomrule
  \end{tabular}
  \caption{
    Comparison of CCNR values in dimensions \(d=3,5,7\): 
    the second column lists the CCNR value of the PPT state maximally violating the Carceller-Tavakoli (PM) witness, and the third gives the CCNR value of the PPT state within the set of $d\times d$ Bell-diagonal states. The PPT state (state No. 6 in Table~\ref{table1}) maximally violating our correlation witness for $d=4$ gives the CCNR value of $1.5$ which is maximal among $4\times 4$ Bell-diagonal PPT states.}
  \label{table0}
\end{table}

\section{Discussion}

We have introduced a family of linear correlation witnesses for a three-party prepare-and-measure scenario with a four-dimensional message space. These witnesses are capable of detecting entanglement in broad families of weakly entangled $4\times 4$ states, as well as in higher-dimensional states with a positive partial transpose (PPT). Tailored to a target entangled state, they perform best for Bloch-diagonal states. Requiring only simple Pauli rotations and product-Pauli measurements, this method provides a versatile and experimentally friendly tool for entanglement detection, robust again realistic noise. In this setting, Main Result~\ref{theorem1}, which represents the central finding of this work, establishes a direct link between the power of the proposed linear correlation witness and the violation of the CCNR criterion.

To test the performance of our approach, we applied Main Result~\ref{theorem1} to show that the introduced prepare-and-measure witness can detect several $4\times 4$ PPT bound entangled states, including both Bloch-diagonal and non-Bloch-diagonal examples, each mixed with varying amounts of isotropic noise. Although the separable bound for our correlation inequality is obtained via a heuristic method, it consistently reproduces the conjectured bound. Moreover, we demonstrated that our prepare-and-measure witness can detect bound entangled states in every finite dimension, from $D = 3$ up to the asymptotic limit $D \to \infty$, and that it outperforms the CCNR criterion for a family of $D\times D$ bound entangled states whenever $D > 4$. However, unlike the Carceller-Tavakoli witness~\cite{Carceller2025PRL}, our construction does not readily generalize to arbitrary message dimensions (such as odd primes or their products), due to its special group-theoretic structures. 

Our scheme and the Carceller-Tavakoli witness schemes share several notable features. Both exhibit strong robustness to noise that well exceeds the thresholds known for the violation of bipartite Bell nonlocality with bound entanglement~\cite{Vertesi2014,Yu2017,Pal2017family}, and therefore making them promising candidates for implementations. In addition, the measurements used in both Ref.~\cite{Carceller2025PRL} and our scheme are separable, with ours being product measurements, and are considerably simpler to implement than entangled measurements. 

However, a key difference in our protocol lies in its simplicity: while the Carceller-Tavakoli witness involves $d$-outcome measurements (for $d$ odd prime), our witness requires only two-outcome measurements to reveal bound entanglement. Another difference concerns the CCNR values of the states optimally violating the witnesses. In particular, we used an iterative algorithm similar to that of Ref.~\cite{Lukacs2022} to maximize the CCNR value over $d\times d$ Bell-diagonal PPT states for $d=3,5,7$, and we compared these maxima to those that achieve the largest violation of the Carceller-Tavakoli witness. As summarized in Table~\ref{table0}, for each tested dimension $d\in\{3,5,7\}$, the Bell-diagonal state that attains the highest CCNR value is distinct from the one that maximally violates the Carceller-Tavakoli witness. This highlights that this witness is less directly tied to the CCNR criterion than our correlation witness.

Notably, for $d=4$, the Bloch-product-diagonal (and hence Bell-diagonal) PPT state (state No. 6 in Table~\ref{table1}) achieving the maximal CCNR value of 1.5 also exactly saturates our prepare-and-measure witness among all PPT states. We further proposed a photonic implementation of this correlation witness which appears to be a promising platform for signaling entanglement in this highly noise-robust bound entangled state. 

Finally, the strong connection we explored between our prepare-and-measure witness and the CCNR criterion for $4\times 4$ Bloch-product-diagonal states (with message dimension four) raises intriguing questions: First, can this link be extended to more general Bloch-diagonal states with higher channel capacity and local state dimensions? Second, might bound entanglement serve as a useful resource in prepare-and-measure protocols, enabling to design novel noise resistant quantum communication tasks?

\subsection*{Data Availability Statement}

The data that support the findings of this study are openly available at the following URL/DOI:
\url{https://github.com/istvanmarton/Classical_bounds_EAPM}.

\subsection*{Acknowledgments}

We acknowledge the support of the EU (CHIST-ERA MoDIC) and the National Research, Development and Innovation Office NKFIH (No.~2023-1.2.1-ERA\_NET-2023-00009 and No.~K145927). T.V. acknowledges support from the `Frontline' Research Excellence Program of the NKFIH (No.~KKP133827). I.M. acknowledges support from the J\'anos Bolyai Research Scholarship of the Hungarian Academy of Sciences.

%

\end{document}